\documentclass[12pt]{article}
\usepackage{amsmath}
\usepackage{amssymb,amsthm}
\usepackage{bbold}
\usepackage{color}

\newtheorem{Thm}{Theorem}[section]
\newtheorem{theorem}[Thm]{Theorem}
\newtheorem{proposition}[Thm]{Proposition}
\newtheorem{corollary}[Thm]{Corollary}

\newtheorem{remark}{Remark}[section]

\newtheorem{definition}[Thm]{Definition}

\title{On the non-perturbative value of the wave function renormalization constant}

\author{Christian D.\ J\"akel\footnote{jaekel@ime.usp.br, Dept.~de Matem\'atica Aplicada, 
Univ. de S\~ao Paulo (USP), Brasil} \ 
and Walter F. Wreszinski\footnote{wreszins@gmail.com, 
Instituto de Fisica, Universidade de S\~ao Paulo (USP), Brazil}}        

\begin{document}

\maketitle

\begin{abstract}
In the perturbative approach to quantum field theory it is common to replace the 
propagator $i (p^{2}-m_{0}^{2}+i\varepsilon )^{-1}$ for a scalar field by
a similar expression, namely $iZ (p^{2}-m^{2}+i\varepsilon )^{-1}$, where 
the shift of the mass from $m_{0}$ to $m$ reflects the mass renormalization and 
the constant~$Z$ is the renormalized field strength (or wave-function). 
We argue that, contrary to general belief, the \emph{nonperturbative value} of~$Z$ 
is \emph{not necessarily} equal to zero in case the two-point function of an 
interacting quantum field theory is, as expected, more 
singular on the light-cone than the corresponding free  field two-point function. 
If, however, (massless) photons or composite (unstable) particles 
are present, the condition $Z=0$ follows from
two qualitatively different 
arguments, one being a theorem due to Buchholz, the other a criterion due to Weinberg.
Hence, the condition $Z=0$ \emph{is}, after all, 
a universal feature of realistic models of elementary particle physics, 
which include massless or unstable particles. 
The results hold within a natural framework which, in the case of gauge theories, 
requires Hilbert space positivity, and therefore the use of non-manifestly covariant gauges.
\end{abstract}

\section{Introduction}

In his recent recollections, 'tHooft (\cite{tHooft}, Sect.~5) emphasizes that 
an asymptotic (divergent) series, such as the power series for the scattering (S) 
matrix in the coupling constant $\alpha = \frac{1}{137}$ in quantum 
electrodynamics ($qed_{1+3}$) does not define a theory rigorously (or mathematically). 
In the same token, Feynman was worried about whether the $qed_{1+3}$ 
S matrix would be ``unitary 
at order 137'' (\cite{Wight2}, discussion  on  p.~126), and in Section 5 of \cite{tHooft}, 'tHooft 
remarks that the uncertainties in the S matrix amplitudes at order~137 for qed are comparable 
to those associated to the ``Landau ghost'' (or pole) \cite{Landau}.  
The situation is considerably worse for strong interactions. 
Except for certain processes accounted for by renormalization group  techniques 
\cite{Weinb3}, the only approach to such fundamental 
issues such as the proton-neutron mass difference~\cite{Bor} or the magnetic moment 
of the neutron \cite{Beane} has remained that of numerical lattice gauge 
theory (LGT).  However, 
LGT is  only  an effective theory in the sense of \cite{Weinb1}, p.~523. 
For  such theories many of the 
pleasant features of quantum field theory, such as invariance or symmetry 
properties, are lost. The 
main reason to believe in quantum field theory (see also \cite{Wight1}) is, therefore, 
in spite of the spectacular success of perturbation theory (\cite{Weinb1}, Chap.~11 
and \cite{Weinb2} Chap.~15), strongly tied 
to non-perturbative approaches. 

This is, in particular, true for quantum 
electrodynamics in the Coulomb gauge in three space dimensions -- $qed_{1+3}$ -- 
for which a recent result \cite{JaWre2} establishes positivity of the (renormalized) 
energy, uniformly in the volume~($V$) and ultraviolet ($\Lambda$) cutoffs. This stability 
result may be an indication of the absence of Landau poles or ghosts in $qed_{1+3}$ 
(see also \cite{Frohlich}). In order to achieve this, the theory in Fock space (for 
fixed values of the cutoffs)  is exchanged for a formulation in which 
(in the words of Lieb and Loss, who were the first to exhibit this phenomenon in a 
relativistic model \cite{LLrel}), ``the electron Hilbert space is linked 
to the photon Hilbert space in an inextricable way''. Thereby, ``dressed photons'' 
and ``dressed electrons'' arise as new entities. This,
required by stability, provides a physical characterization of the
otherwise only mathematically motivated non-Fock representations which arise in an
interacting theory (\cite{Wight},\cite{Wight1}), and forms the backbone of the constructive
part of our approach, whose final aim is to prove the axioms of the framework expounded
in section 3. The non-unitary character of the transformations to the physical Hilbert space,
when the space and ultraviolet cutoffs are removed, may explain the absence of a ``universal''
high-energy behaviour sought by Landau, when expressing his doubts about the consistency
of relativistic quantum field theory (see the discussion in~\cite{Wight1}).   

Although Weinberg 
(\cite{Weinb2}, Chap.~18, p.~136) stated that  ``there is 
today a widespread view that interacting quantum field theories that are 
not asymptotically free, like quantum electrodynamics, are not mathematically 
consistent'', we tend to disagree. 
The renormalization group arguments, which lead to the conclusion stated by Weinberg
ignore the non-unitary transformations mentioned in the last paragraph, they  
are part of the mathematical structure of relativistic quantum field theories 
(non-Fock representations) which is simply \emph{not accounted} for by the 
Standard arguments. As, strictly speaking, 
the renormalization group arguments are not conclusive, as Weinberg himself 
remarks (ibid, p.~137), we argue that our approach at least carries the 
benefit of the doubt.

\section{The Field Algebra}

In order to formulate the above-mentioned phenomenon, we consider a 
\emph{field algebra} ${\cal F}^{0} \equiv \{A_{\mu}^{0},\psi^{0},\bar{\psi^{0}}\}$ (containing 
the identity) generated by
the free vector potential $A_{\mu}^{0}, \mu=0,1,2,3$, and the 
electron-positron fields~$\psi^{0},\bar{\psi^{0}}$. 
We may assume that the field algebra is initially defined 
on the Fock-Krein (in general, indefinite-metric, see \cite{Bognar}) tensor product of 
photon and fermion Fock spaces. However, 
of primary concern for us will be an \emph{inequivalent} 
representation of the field algebra  ${\cal F}^{0}$ on a (physical) 
Hilbert space ${\cal H}$, with generator of time-translations - the physical 
Hamiltonian $H$ - satisfying \emph{positivity}, \emph{i.e.}, 
	\begin{equation}
		H \ge 0 \; ,
		\label{(1)}
	\end{equation}
and such that
	\begin{equation}
		H \Omega = 0  \; ,
		\label{(2)}
	\end{equation}
where $\Omega \in {\cal H}$ is the vacuum vector, and
	\begin{align}
		A_{\mu} &\equiv A_{\mu}(A_{\mu}^{0},\psi^{0}, \bar{\psi^{0}})  \; ,
		\label{(3)}
		\\
		\Psi & \equiv \Psi(A_{\mu}^{0},\psi^{0}, \bar{\psi^{0}})  \; ,
		\label{(4.1)}
		\\
		\bar{\Psi} & \equiv \bar{\Psi}(A_{\mu}^{0},\psi^{0}, \bar{\psi^{0}})  \; .
		\label{(4.2)}
	\end{align}  

\subsection{Gauge transformations}
We may define, as usual, the (restricted) class of c-number $U(1)$ local 
gauge transformations, acting on ${\cal F}^{0}$, by the maps
	\begin{align}
		A_{\mu}^{0}(f) & \to A_{\mu}^{0}(f) +c \, \langle f,\partial^{\mu}u \rangle  \; ,
		\label{(5.1)}
		\\
		\psi^{0}(f) & \to \psi^{0} \bigl( {\rm e}^{iu} f \bigr)  \; ,
		\label{(5.2)}
		\\
		\bar{\psi^{0}}(f) & \to \bar{\psi^{0}} \bigl( {\rm e}^{-iu} f \bigr)  \; ,
		\qquad f \in {\cal S}(\mathbb{R}^{1+s}) \; . 
		\label{(5.3)}
	\end{align}
In qed, $c=\frac{1}{e}$ and $u$ satisfies certain regularity conditions, which guarantee that 
	\[
		{\rm e}^{\pm iu}f \in {\cal S}(\mathbb{R}^{1+s}) \quad \text{if} \; f \in {\cal S}(\mathbb{R}^{1+s}) \; , 
		\quad \text{and} \quad  | \langle f,\partial^{\mu}u \rangle | < \infty \; . 
	\]
Here, ${\cal S}$ denotes  Schwartz space (see, \emph{e.g.}, \cite{BB}), 
$\langle \, . \, , \, . \, \rangle$  denotes  the $L^{2}(\mathbb{R}^{1+s})$ scalar product 
and $s$ is the space dimension. Such transformations have been considered in a 
quantum context, that of
(massless) relativistic qed in two  
space-time dimensions (the Schwinger model) by Raina and Wanders,
but their unitary implementability is a delicate matter \cite{RaWa}. 
We shall use (6)-(8) merely as
as a guiding principle to construct the observable algebra, to which we now turn. 

The \emph{observable algebra} is assumed to consist of 
gauge-invariant objects, namely the fields
	\begin{equation}
		F_{\mu,\nu} = \partial_{\mu} A_{\nu}-\partial_{\nu}A_{\mu}  \; ,
		\label{(6)}
	\end{equation}
with $A_{\mu}$ given by (3), describing the dressed photons, and the quantities (10)
below. In order to define them, we assume the existence of gauge-invariant
quantities $\Psi,\bar{\Psi}$ in \eqref{(4.1)}, \eqref{(4.2)}, which create-destroy 
electrons-positrons ``with their photon clouds''. 

While we hope that the results in \cite{JaWre2} will eventually lead to an (implicit or explicit) 
expression for $\Psi,\bar{\Psi}$, it should be emphasized that this is a very difficult, 
open problem; see the important work of Steinmann in perturbation theory \cite{Steinmann1}. 
We also note that we shall only consider the \emph{vacuum sector}, 
\emph{i.e.}, and therefore the fermion part of the observable algebra will be assumed to consist of the combinations
	\begin{align}
		A(f,g) & \equiv \bar{\Psi}(f)\Psi(g) 
		\mbox{ with } f,g \in {\cal S}(\mathbb{R}^{1+s})  \; ,
		\nonumber
		\\
		B(f,g) & \equiv \Psi(f)\bar{\Psi}(g) 
		\mbox{ with } f,g \in {\cal S}(\mathbb{R}^{1+s})  \; .
		\label{(7)}
	\end{align}
The existence of \emph{charged sectors} is a related 
open problem, which will not concern us.  

\section{A framework for relativistic quantum gauge theories}

An appealing framework for relativistic quantum gauge theories based 
on a new concept of string localization,
but excluding indefinite metric and ``ghosts'', has been recently 
implemented by Schroer, Mund, and
Yngvason. For references to the literature on this topic, we refer to 
Bert Schroer's recent preprint \cite{Schrn}.
The framework we will use is due to Lowenstein and 
Swieca \cite{LSwi}. Raina and Wanders \cite{RaWa} have used it to 
construct a theory of $qed_{1+1}$, the Schwinger model.
It has one important point in common with the
string-localisation theory: the absence of ``ghosts''. 

The theory will be defined by its \emph{n-point Wightman functions} \cite{StreWight} of 
observable fields. Alternatively, a Haag-Kastler theory~\cite{HK} 
may be envisaged. It has been shown in the seminal work of the latter authors that 
the whole content of a theory can be expressed in terms of its observable algebra. 
In the case of gauge theories, the latter corresponds to the algebra generated by gauge 
invariant quantities \cite{StrWight}. As remarked by Lowenstein and Swieca \cite{LSwi}, 
the observable algebra, being gauge-invariant, should have the same representations, 
independently of the gauge of the field algebra it is constructed from. Thus, $n$-point 
functions constructed over the observable algebra should be the same in all gauges. 
These remarks fully justify the  usage  
of non-covariant gauges, which, as we shall see, are of particular importance in a 
non-perturbative framework. For scattering theory 
and particle concepts within a theory of local observables, see \cite{ArHa}, \cite{BuPoSt}, 
and \cite{BuchSum} for a lucid review.

There exist various arguments supporting the use of non-covariant gauges in 
relativistic quantum field theory: they are of both physical and mathematical nature. 
In part one of his treatise, Weinberg notes (\cite{Weinb1}, p.~375, Ref.~2): ``the use of 
Coulomb gauge in electrodynamics was strongly advocated by Schwinger \cite{Schw1} 
on pretty much the same grounds as here: that we ought not to introduce photons with 
helicities other than $\pm 1$''. Indeed, as shown by Strocchi \cite{Strocchi}, a framework 
excluding ``ghosts'' necessarily requires the use of non-manifestly covariant gauges, 
such as the Coulomb gauge in $qed_{1+3}$, the Weinberg or unitary gauge in the Abelian 
Higgs model \cite{Weinb2},  and the Dirac \cite{Dirac} or light-cone gauge in quantum 
chromodynamics \cite{SriBro}. Another instance of the physical-mathematical advantage 
of a non-covariant gauge is the ``$\alpha=\sqrt{\pi}$'' gauge in massless $qed_{1+1}$ - the 
Schwinger model \cite{Schw} - see \cite{LSwi}, \cite{RaWa}. As in the Coulomb gauge 
in $qed_{1+3}$, there is no need for indefinite metric in this gauge, \emph{i.e.}, the zero-mass 
longitudinal part of the current is gauged away, and one has a solution of Maxwell's 
equations (as an operator-valued, distributional identity)
	\begin{equation}
		\partial_{\nu}F^{\mu,\nu}(x) = -ej^{\mu}(x)
		\label{(8)}
\end{equation}
on the whole Hilbert space. This is an important ingredient 
in Buchholz's theorem \cite{Buch}, to which we come back in the sequel. 

The structure of the observable algebra is quite simple in  the Coulomb 
gauge: the field \eqref{(6)} is 
just the electric field, which is defined in terms of a massive scalar field, the 
quantities \eqref{(7)} are, in this gauge, rigorous versions of the (path-dependent) quantities
	\begin{equation}
		\psi(x) {\rm e}^{ ie \int_{x}^{y} {\rm d} t_{\mu} \; A^{\mu}(t)} \psi^{*}(y)
\end{equation}
and  their adjoints (in the distributional sense), see \cite{LSwi} and \cite{RaWa}. 
In the case of $qed_{1+3}$, such quantities are plagued by infrared divergences, 
see the discussion in \cite{Steinmann1}. As a consequence of  the simple structure of 
the observable algebra, one arrives at a correct physical-mathematical picture 
of spontaneous symmetry breakdown (\cite{LSwi},\cite{RaWa}); in covariant 
gauges this picture is masked by the presence of spurious gauge excitations.

In \cite{StreWight}, pp.~107-110, it is shown that if the (n-point) \emph{Wightman functions} 
satisfy 
\begin{itemize}
\item [$a.)$] the relativistic transformation law; 
\item [$b.)$] the spectral condition; 
\item [$c.)$] hermiticity; 
\item [$d.)$] local commutativity; 
\item [$e.)$] positive-definiteness, 
\end{itemize}
then they are the \emph{vacuum expectation values of a field theory} satisfying 
the Wightman axioms, except, eventually,  the uniqueness of 
the vacuum  state. We refer to \cite{StreWight} or \cite{RSII}  for 
an account of Wightman theory, and for the description of these properties.  It has been shown
in \cite{LSwi}, \cite{RaWa} that 
$qed_{1+1}$ in the ``$\alpha=\sqrt{\pi}$'' gauge satisfies $a.)-e.)$. 
The crucial positive-definiteness condition~e.) has been 
shown in \cite{LSwi} to be a consequence of the positive-definiteness of a subclass 
of the n-point functions of the Thirring model~\cite{Thirrm} in the formulation of 
Klaiber~\cite{Klai}. Positive-definiteness of the Klaiber n-point functions was 
rigorously proved by Carey, Ruijsenaars and Wright \cite{CRW}. Uniqueness of the 
vacuum holds in each irreducible subspace of the (physical) Hilbert 
space~${\cal H}$~\cite{RaWa}, as a result of the cluster property; see also \cite{LSwi}.

For $qed_{1+3}$ in the Coulomb gauge, we shall \emph{assume} $a.)-e.)$ for the 
$n$-point functions of the observable  fields. We conjecture that this framework 
is also adequate for other relativistic quantum gauge field theories, as previously 
discussed. Concerning the case of $qed_{1+3}$ in the Coulomb gauge, for both 
the electron and the photon propagators in perturbation theory, dynamics ``restores'' 
Lorentz covariance, \emph{i.e.}, the instantaneous Coulomb interaction and the transverse 
part combine to yield a Lorentz invariant expression (see~\cite{JJS}, Sections 4-4 
and 4-6). Assuming Lorentz covariance, positivity of the energy (1) in the physical 
Hilbert space ${\cal H}$ yields the spectral condition~b.). The crucial mathematical 
reason for choosing a non-covariant gauge is, as we shall see, the positive-definiteness 
condition e.). Concerning the uniqueness of the vacuum, we shall assume it is valid by 
restriction to an irreducible component of ${\cal H}$, as in $qed_{1+1}$.

We hope that the new constructive approach mentioned in the introduction may be
adequate to prove the above axioms. There are hopes, so far, to  arrive at the positivity condition (1),
which, together with a.), yields the spectral condition b.), but the details of the removal
of the cutoffs remain to be studied. The origin of the constructive
approach remains, however, the same as Wightman's, namely, canonical quantization of classical 
Lagrangian field theory, by necessity with cutoffs, which are, however, subsequently removed
(see \cite{Wight},\cite{Wight2} for introductions).

\subsection{The K\"all\'en-Lehmann representation}

A major dynamical issue in quantum field theory is the (LSZ or Araki-Haag) asymptotic 
condition (see, \emph{e.g.}, \cite{BuchSum} and references given there), which 
relates the theory, whose objects are the fundamental observable fields, to \emph{particles}, 
described by physical parameters (\emph{mass} and \emph{charge} in qed). This issue is equivalent 
to the renormalization (or normalization) of perturbative quantum field theory, which 
itself is related to the construction of continuous linear extensions of certain functionals, 
such as to yield well-defined tempered 
distributions (see \cite{Scharf} and references 
given there). On the other hand, in a non-perturbative framework, a theory of renormalization 
of masses and fields also exists, and ``has nothing directly to do with the presence of 
infinities'' (\cite{Weinb1}, p.~441, Sect.~10.3). We adopt a related proposal, which we formulate 
here, for simplicity, for a theory of a self-interacting scalar field $A$ of mass $m$ satisfying 
the Wightman axioms (modifications are mentioned in the sequel).  
We assume that  $A$ is an operator-valued 
tempered distribution on the Schwartz space ${\cal S}$ (see \cite{RSII}, Ch.~IX). 

\bigskip
We have the following  result, concerning the spectral representation of the 
two-point function $W_{2}$ (\cite{RSII}, p.~70, Theorem IX-34):

\begin{theorem}[The K\"{a}ll\'{e}n-Lehmann representation]
\label{th:2} 
	\begin{equation}
		W_{2}^{m}(x-y)= \langle \Omega,A(x)A(y)\Omega \rangle
		= \frac{1}{i} \int_{0}^{\infty} {\rm d}\rho(m_\circ^{2}) \; 
		\Delta_{+}^{m_\circ}(x-y) \; , 
		\label{(9)}
\end{equation}
where $\Omega$ denotes the vacuum vector, $x=(x_{0},\vec{x})$, and
	\begin{align}
		\Delta_{+}^{m_\circ}(x) =\frac{i}{2(2\pi)^{3}} 
		\int_{\mathbb{R}^{3}} {\rm d}^{3}\vec{k} \; 
		\frac{ {\rm e}^{-ix_{0}\sqrt{m_\circ^{2}+\vec{k}^{2}}
		+i\vec{x}\cdot\vec{k}}}{\sqrt{m_\circ^{2}+\vec{k}^{2}}} 
		\label{(10)}
	\end{align}
is the two-point function of the free scalar field of mass $m_\circ$, and $\rho$ is 
a polynomially-bounded measure on $[0,\infty)$, \emph{i.e.},
	\begin{equation}
		\int_{0}^{L} {\rm d} \rho(m_\circ^{2}) \le C(1+L^{N})
		\label{(11)}
\end{equation}
for some constants $C$ and $N$. It is further assumed that
	\begin{equation}
		\langle \Omega,A(f)\Omega \rangle = 0 \qquad \forall  f \in {\cal S} \; . 
		\label{(12)}
\end{equation}
\end{theorem}

Note that \eqref{(9)} is symbolic; for its proper meaning, which relies on \eqref{(11)},
see \cite{RSII}.

\begin{proposition} For a scalar field of mass $m \ge 0$, the 
measure ${\rm d}\rho(m_\circ^{2}) $ 
appearing in the K\"{a}ll\'{e}n-Lehmann spectral representation 
allows a decomposition
\label{prop:1} 
	\begin{equation}
		{\rm d}\rho(m_\circ^{2}) = Z\delta(m_\circ^{2}-m^{2}) + d\sigma(m_\circ^{2}) \; , 
		\label{(13)}
\end{equation}
where
	\begin{equation}
		0 < Z < \infty
		\label{(14)}
\end{equation}
and
	\begin{equation}
		\int_{0}^{L} {\rm d} \sigma(a^{2}) \le C_{1}(1+L^{N_{1}})
		\label{(15)}
\end{equation}
for some constants $C_{1}$ and $N_{1}$. 
\end{proposition} 

\begin{proof} 
By the Lebesgue decomposition (see, \emph{e.g.}, Theorems I.13, I.14, p.~22 of \cite{RSI})
	\begin{equation}
		{\rm d}\rho = {\rm d}\rho_{p.p.} + {\rm d}\rho_{s.c.} + {\rm d}\rho_{a.c.} \; , 
		\label{(16.1)}
\end{equation}
where p.p.~denotes the \emph{pure point}, s.c.~denotes 
 the \emph{singular continuous} and 
a.c.~denotes 
 the \emph{absolutely continuous} parts of ${\rm d}\rho$. 
Since there exists only a finite number of masses
in nature, there is no accumulation point of the values of the
mass in a realistic theory, and the pure point part of the measure is, in fact, discrete. 
This is, however, mathematically speaking, an assumption, as well as the 
identification of the masses in the corresponding terms in (21), (29) and (30) below,  
with those associated to the scalar, photon and electron-positron fields, 
respectively. This assumption is verified in the case of free fields, but there 
is one important difference: in the interacting case, the masses in (21), (29) 
and (30) are interpreted as renormalized masses and may differ from 
the ``bare" masses, \emph{i.e.}, those originally occurring in the fields 
obtained through canonical quantization in Fock space with cutoffs.
For a scalar field of mass $m$, we obtain
         \begin{equation}
		{\rm d}\rho_{p.p.}(m_\circ^{2}) = Z\delta(m_\circ^{2}-m^{2}) \; ,
		\label{(16.2)}
\end{equation}
where $Z$ satisfies \eqref{(14)}, and, by \eqref{(11)}, 
\eqref{(13)} and \eqref{(16.1)}, ${\rm d}\sigma$ satisfies \eqref{(15)}. 
\end{proof}

\begin{remark}
\label{Remark 0.1}
Of course, $Z=0$ in \eqref{(16.2)}, if there is no discrete component of mass $m$ in the total 
mass spectrum of the theory. In general $Z$ in each discrete component of $d\rho$ has only to satisfy
         \begin{equation}
               0 \le Z < \infty
               \label{(14.1)} \; , 
\end{equation}  
because of the positive-definiteness conditions e.) (or the positive-definite 
Hilbert space metric).
\end{remark}
 
\begin{remark}
\label{Remark 0.2}
Expression \eqref{(16.2)} corresponds precisely to (\cite{Weinb1}, p.~461, Equ.~(10.7.20)). 
Thus, Proposition~\ref{prop:1} is just a mathematical statement of the nonperturbative 
renormalization theory, as formulated by Weinberg. Thus, 
the physical interpretation of $Z$ is that $0< Z < \infty$ 
is the field strength renormalization constant, due to the 
fact that the physical field $A_{phys}$ is normalized (or renormalized) by the 
one-particle condition (\cite{Weinb1}, (10.3.6)) which stems from the LSZ (or 
Haag-Ruelle) asymptotic condition (see also Sections~2 and~5 of \cite{BuchSum} 
and references given there for the appropriate assumptions). A general field, 
as considered in (9), does not have this normalization. By the same token, 
the quantity $m^{2}$ in \eqref{(16.2)} is interpreted as the physical (or renormalized) 
mass associated to the scalar field.  
No further conditions are expected to be imposable on quantum fields and, indeed, 
we shall see that the additional requirement that the interacting fields satisfy the 
equal-time commutation relations (ETCR) leads to an \emph{entirely different picture}. 
As we shall see, it is the latter picture which seems to agree with the general 
belief (\cite{Weinb1}, \cite{Barton1}, \cite{Lehmann}, \cite{Kallen}, 
\cite{Haag}, \cite{Wight}) that the condition
	\begin{equation}
		Z=0
		\label{(17)}
	\end{equation}
is a \emph{general} condition in relativistic quantum field theory (rqft). It is immediate, 
however, that this proposed universality of \eqref{(17)} contradicts the definition of $Z$ given 
by \eqref{(13)} of Proposition~\ref{prop:1}, because the latter implies the absence of 
one-particle hyperboloids. Their presence is, however, the central building block of 
scattering theory \cite{BuchSum}.  
\end{remark}

For the purposes of identification with Lagrangian field theory, one 
may equate the $A(.)$ of \eqref{(9)} with the ``bare'' scalar 
field $\phi_{B}$ (\cite{Weinb1}, p.~439), whereby
	\begin{equation}
		A = \sqrt{Z} A_{phys} \; . 
		\label{(18.1)}
	\end{equation}
Similarly, for the electron and photon fields, conventionally,
	\begin{align}
		\Psi & = \sqrt{Z_{2}} \,  \Psi_{phys} \; , 
		\label{(18.2)}
		\\
		F_{\mu,\nu} & = \sqrt{Z_{3}} \, F_{\mu,\nu, \, phys} \; . 
		\label{(18.3)}
	\end{align}
Above, we refer to the previously discussed fields in \eqref{(4.1)}, \eqref{(4.2)} 
and \eqref{(6)}.
 
For $F_{\mu,\nu}$~and $\Psi$, we have the analogues of \eqref{(9)}, namely 
	\begin{align}
		& \langle F_{\mu,\nu}(x) \Omega, F_{\mu,\nu}(y)\Omega \rangle 
		= \int {\rm d} \rho_{ph}(m_\circ^{2})\int\frac{{\rm d}^{3}p}{2p_{0}} 
		(-p_{\mu}^{2}g_{\nu\nu}-p_{\nu}^{2}g_{\mu\mu}) {\rm e}^{ip \cdot (x-y)} \; , 
		\label{(19.1)}
	\end{align}
with $\mu \ne \nu$, no summations involved, and $p_{0}=\sqrt{ {\vec{p}\,}^{2}+m_\circ^{2}}$. 
Denoting spinor indices by $\alpha,\beta$, we have
	\begin{align}
		S_{\alpha,\beta}^{+}(x-y) 
		& = \langle \Omega, \Psi_{\alpha}(x)\bar{\Psi}_{\beta}(y) \Omega \rangle
		\label{(20.1)} \\
		& = \int_{0}^{\infty} {\rm d}\rho_{1}(m_\circ^{2}) S_{\alpha,\beta}^{+}(x-y;m_\circ^{2}) 
		+ \delta_{\alpha,\beta} \int_{0}^{\infty} {\rm d} \rho_{2}(m_\circ^{2}) 
		\Delta^{+}(x-y;m_\circ^{2}) \, , 
		\nonumber
	\end{align}
with ${\rm d} \rho_{ph}, {\rm d} \rho_{1}, {\rm d} \rho_{2}$ positive, polynomially bounded 
measures, and $\rho_{1}$ satisfying certain bounds with respect 
to $\rho_{2}$ (\cite{Lehmann}, p.~350). Again, as in \eqref{(16.2)},
	\begin{equation}
		{\rm d} \rho_{ph}(m_\circ^{2}) = Z_{3} \delta (m_\circ^{2}) 
		+ {\rm d} \sigma_{ph}(m_\circ^{2})
		\label{(19.2)}
	\end{equation}
and
	\begin{equation}
		{\rm d} \rho_{1}(m_\circ^{2}) = Z_{2}\delta(m_\circ^{2}-m_{e}^{2}) 
		+ {\rm d} \sigma_{1}(m_\circ^{2}) \; , 
		\label{(20.2)}
	\end{equation}
with $m_{e}$ the renormalized electron mass, according to conventional notation. 

\goodbreak 
We have, by the general condition \eqref{(14.1)},
	\begin{align}
		0 \le Z_{3} < \infty \; , 
		\label{(19.3)}
		\\
		0 \le Z_{2} < \infty \; . 
		\label{(20.3)}
	\end{align}
When $Z_{3} > 0$, the renormalized electron charge follows from 
(\cite{Weinb1}, (10.4.18)). Assumption \eqref{(12)}, which is also expected 
to be generally true on physical grounds, becomes
	\begin{align}
		\langle \Omega, F_{\mu,\nu}(f) \Omega \rangle & = 0 
		\quad \forall f \in {\cal S} \; , 
		\label{(19.4)}
		\\
		\langle \Omega, \Psi_{\alpha}(f) \Omega \rangle & = 0 \quad 
		\forall  f \in {\cal S} \mbox{ and } \alpha 
		\text{  a spinor index} \, . 
		\label{(20.4)}
	\end{align}

In summary, Proposition~\ref{prop:1} provides a rigorous (non-perturbative) definition 
of the field-strength (or wave-function) renormalization constant. In the proof 
of Theorem~\ref{th:2} (\cite{RSII}, p.~70), the positive-definiteness condition 
e.) plays a major role. Thus, the definition of $Z$ and its range \eqref{(14.1)} (which 
depends on the positivity of the measure ${\rm d}\rho$) strongly hinge on the fact that 
the underlying Hilbert space has a positive metric. Parenthetically, the 
positive-definiteness condition on the Wightman functions is ``beyond the powers 
of perturbation theory'', as Steinmann aptly observes \cite{Steinmann1}.

\section{The Singularity Hypothesis}

If $0<Z<\infty$ \eqref{(14)} holds,  
the assumption of ETCR for the physical 
fields may be written (in the distributional sense)
		\begin{equation}
			\left[ \frac{\partial A_{phys}(x_{0},\vec{x}\,)}{\partial x_{0}},
			A_{phys}(x_{0},\vec{y}\,) \right] = -\frac{i}{Z} \; \delta(\vec{x}-\vec{y}\,) \; . 
			\label{(21)}
		\end{equation}
Together with \eqref{(9)} and \eqref{(18.1)}, \eqref{(21)} yields (see also \cite{Lehmann},
\cite{Kallen}, \cite{Barton1}, \cite{Weinb1})
	\begin{equation}
		\frac{1}{Z} = \int_{0}^{\infty} {\rm d}\rho(m_\circ^{2}) \; . 
		\label{(22)}
	\end{equation} 
Together with \eqref{(13)}, \eqref{(21)} would yield
	\begin{equation}
		\frac{1}{Z} = Z + \int_{0}^{\infty} {\rm d}\sigma(m_\circ^{2}) \; . 
		\label{(23.1)}
	\end{equation}
Since $d\sigma$ is a positive measure, we obtain from \eqref{(23.1)} the inequality 
	\begin{equation}
		Z \le 1
		\label{(23.2)}
	\end{equation}
(\cite{Weinb1}, p.~361, \cite{Barton1}, Equ.~(9.19)).

Formulas \eqref{(21)} and \eqref{(22)} have been extensively used as a heuristic guide, 
even by two of the great founders of axiomatic (or general) quantum 
field theory, Wightman and Haag, to substantiate the previously mentioned 
conjecture that $Z=0$ is expected to be a \emph{general} condition for interacting 
fields. One reason is that \eqref{(21)} suggests that, if $Z=0$ holds, the light-cone 
singularity of the two-point function of interacting fields is expected to 
be stronger than the one exhibited by the free field - we shall refer to 
this assumption, shortly, as the ``singularity hypothesis''. Indeed, 
in \cite{Wight}, p.~201, it is observed that ``$\int_{0}^{\infty} {\rm d}\rho(m_\circ^{2}) = \infty$ 
is what is usually meant by the statement that the field-strength renormalization 
is infinite''. This follows from \eqref{(22)}, with ``field-strength renormalization'' interpreted 
as $\frac{1}{Z}$. The connection with the singularity hypothesis comes 
next (\cite{Wight}, p.~201), with the observation that, by \eqref{(9)}, $W_{2}$ will 
have the same singularity, as $(x-y)^{2}=0$, as does $\Delta_{+}(x-y;m^{2})$ 
if $\int_{0}^{\infty} {\rm d}\rho(m_\circ^{2}) < \infty$. 

\section{Steinmann Scaling Degree}

In order to formulate the singularity hypothesis in rigorous terms, we recall the 
Steinmann scaling degree $sd$ of a distribution \cite{Steinmann}; for 
a distribution $u \in {\cal S}^{'}(\mathbb{R}^{n})$, let $u_{\lambda}$ denote 
the ``scaled distribution'', defined by 
	\[
		u_{\lambda}(f) \equiv u(f(\lambda^{-1} \cdot)) \; . 
	\]
As $\lambda \to 0$, we expect that $u_{\lambda} \approx \lambda^{-\omega}$ for 
some $\omega$, the ``degree of singularity'' of the distribution $u$. Hence, we set
	\begin{equation}
		sd(u) \equiv \inf \, \bigl\{\omega \in \mathbb{R} \mid \lim_{\lambda \to 0} 
		\lambda^{\omega} u_{\lambda} = 0 \bigr\} \; , 
		\label{(24.1)}
	\end{equation}
with the proviso that if there is no $\omega$ 
satisfying the limiting condition above, 
we set $sd(u) = \infty$. For the free scalar field of mass $m \ge 0$, it is 
straightforward to show from the explicit form of the two-point function in terms of 
modified Bessel functions that
	\begin{equation}
		sd(\Delta_{+}) = 2 \; . 
		\label{(24.2)}
	\end{equation}
In \eqref{(24.2)}, and the forthcoming equations, we omit the mass superscript. 
From Theorem~\ref{th:2}, we have that for $f \in {\cal S}(\mathbb{R}^{4})$ the 
interacting two-point function satisfies
        \begin{align}
		W_{+}(f) = \int_{0}^{\infty}{\rm d}\rho(m_\circ^{2})\int_{\mathbb{R}^{3}}  
		\frac{d\vec{p}}{\sqrt{{\vec{p}\,}^{2}+m_\circ^{2}}} \; 
		\tilde{f} \left(\sqrt{{\vec{p}\, }^{2}+m_\circ^{2}},\vec{p}\, ) \right) \; . 
		\label{(24.3)}
	\end{align}
Here $\tilde{f} \in {\cal S}(\mathbb{R}^{4})$ denotes the Fourier transform of $f$.  

\begin{definition}
\label{def:1} 
We say that the \emph{singularity hypothesis} holds for an interacting scalar field if 
	\begin{equation}
		sd(W_{+}) > 2 \; . 
		\label{(24.4)}
	\end{equation}
\end{definition}

\begin{proposition}
\label{prop:2} If the total spectral mass is finite, \emph{i.e.},
	\begin{equation}
		\int_{0}^{\infty} {\rm d}\rho(a^{2}) < \infty \; , 
		\label{(24.5)}
	\end{equation}
then
	\begin{equation}
		sd(W_{+}) \le 2 \; ; 
		\label{(24.6)}
	\end{equation}
\emph{i.e.}, the scaling degree of $W_{+}$ cannot be strictly greater than 
that of a free theory, and thus, by Definition~\ref{def:1}, the singularity 
hypothesis~\eqref{(24.4)} is \emph{not} satisfied.
\end{proposition}

\begin{proof} 
The scaled distribution corresponding to $W_{+}$ is given by
	\begin{align}
		& W_{+,\lambda}(f) = \lambda^{-2}\int_{0}^{\infty}
		{\rm d}\rho(m_\circ^{2}) \int_{\mathbb{R}^{3}} 
		\frac{d\vec{p}}{\sqrt{{\vec{p}\,}^{2}+\lambda^{2}m_\circ^{2}}} \, 
		\tilde{f} \left( \sqrt{{\vec{p}\,}^{2}
		+\lambda^{2}m_\circ^{2}},\vec{p}\, \right) \; . 
		\label{(24.7)}
	\end{align} 
Assume the contrary to \eqref{(24.6)}, \emph{i.e.}, that $sd(W_{+})=\omega_{0}>2$. 
Then, by the definition of the $sd$, if $\omega < \omega_{0}$, one must have
	\begin{equation}
		\lim_{\lambda \to 0} \lambda^{\omega} W_{+,\lambda}(f) \ne 0 \; . 
		\label{(24.8)}
	\end{equation}
Choosing 
	\begin{equation}
		\omega = \omega_{0}-\delta >2
		\label{(24.9)}
	\end{equation}
in \eqref{(24.8)}, we obtain from \eqref{(24.7)} and \eqref{(24.8)} that
	\begin{align}
		\lim_{\lambda \to 0} \lambda^{\omega-2} 
		\left( \int_{0}^{\infty}{\rm d}\rho(m_\circ^{2}) \int_{\mathbb{R}^{3}}
		 \frac{d\vec{p} }{\sqrt{{\vec{p}\,}^{2}+\lambda^{2}m_\circ^{2}}}
		 \tilde{f}\left(\sqrt{{\vec{p}\,}^{2}+\lambda^{2}m_\circ^{2}},\vec{p}\, \right) 
		 \right)\ne 0 \; . 
		\label{(24.10)}
	\end{align}
The limit, as $\lambda \to 0$, of the term inside the brackets in \eqref{(24.10)}, 
is readily seen to be finite by the Lebesgue dominated convergence 
theorem due to the assumption \eqref{(24.5)} and the fact that 
$\tilde{f} \in {\cal S}(\mathbb{R}^{4})$; but this contradicts \eqref{(24.8)} 
because of \eqref{(24.9)}.
\end{proof}

\begin{remark}
\label{Remark-1} 
Proposition~\ref{prop:2} holds in dimension $n=2$ by the same proof. 
For dimension $n=1$, the proof shows that we must replace 
${\cal S}^{'}(\mathbb{R})$ by ${\cal S}^{'}_{0}(\mathbb{R}) 
\equiv \{f \in {\cal S}(\mathbb{R}) \mid \tilde{f}(0)=0\}$.
\end{remark}

Proposition~\ref{prop:2} makes Wightman's previously mentioned remark precise, 
and suggests, together with \eqref{(22)} the following expectation:

\paragraph{Conjecture 1} The singularity 
hypothesis (Definition~\ref{def:1}) implies 
 $Z= 0$; alternatively, $Z= 0$ is a necessary condition for the singularity hypothesis.

\bigskip
Conjecture 1 is also stated in slightly different 
words by Haag (\cite{Haag}, p.~55), 
who remarks: ``In the renormalized perturbation expansion one relates formally 
the true field $A_{phys}$ to the canonical field $A$ (our notation) which satisfies 
\eqref{(18.1)}, where $Z$ is a constant (in fact, zero). This means that the fields 
in an interacting theory are more singular objects than in the free theory, and we 
do not have the ETCR.''

\bigskip
If we use the rigorous definition of $Z$ (Proposition~\ref{prop:1}), we find:

\begin{theorem}
\label{thm:3} 
\quad 
\begin{itemize}
\item [$a.)$] If the ETCR is not assumed,  
then $sd(W_{+}) > 2$ is possible, no matter whether there is a discrete 
contribution to the K\"{a}ll\'{e}n-Lehmann  measure or not.
This means  Conjecture 1 is false; 
\item [$b.)$] If the ETCR is assumed, 
then $sd(W_{+}) > 2$ implies that there is no discrete contribution 
to the K\"{a}ll\'{e}n-Lehmann  measure. Hence Conjecture~1  is true.
\end{itemize}
\end{theorem}

\goodbreak
\begin{proof} 
\quad
\begin{itemize}
\item [$a.)$] By \eqref{(13)}, $\int_{0}^{\infty} {\rm d}\rho(m_\circ^{2}) = \infty$ holds whatever 
value of $Z$ satisfying \eqref{(14.1)} is chosen --- in particular, $Z=0$ --- whenever 
$\int_{0}^{\infty} d\sigma(m_\circ^{2}) = \infty$; 
\item [$b.)$] If the ETCR is assumed, \eqref{(22)} holds 
whenever $0<Z<\infty$, in which case $\int_{0}^{\infty} {\rm d}\rho(m_\circ^{2}) < \infty$. Thus, 
by \eqref{(14.1)}, $Z=0$ is the only possibility to 
render $\int_{0}^{\infty} {\rm d}\rho(m_\circ^{2}) = \infty$. 
\end{itemize}
\end{proof}

\begin{remark}
One explicit example of $a.)$~above is furnished by a generalised free field 
with a non-integrable mass spectrum but a delta function at a mass shell. 
We thank Erhard Seiler for this remark.
\end{remark}

\begin{remark}
For $Z=0$ the r.h.s. of (35) is mathematically meaningless, which 
is an indication that the ETCR may
not hold (although, of course, no proof of this fact). We did not, 
however, assume $Z=0$ in b.) above, but
rather the ETCR (35) for $0<Z<\infty$, which is the only mathematically 
meaningful statement. This assumption 
yields (36) --- which is \emph{not} generally valid --- from which, together with (22), 
it follows that
the excluded value $Z=0$ is the only one which may lead to infinite 
spectral mass. There is,therefore,
no self-contradiction in statement b.).

On the other hand, for $Z=0$, nothing can be really said about (36)!. 
This is carefully demonstrated in page (863) of
Wightman's famous article \cite{Wight3}, the same in which he 
introduced the Wightman axioms: the
$''\infty \times \delta''$ term is definitely misleading (see \cite{Wight3}, 
top of page 864). But, if such
is the case, what is the basis of the singularity hypothesis?

Actually there is independent reason to believe that the singularity 
hypothesis is generally true in relativistic
quantum field theories, from the (up to the present nonrigorous) 
theory of the renormalization group: even 
in the (believed to be) asymptotic free quantum chromodynamics, 
the critical exponents remain anomalous
(see \cite{Weinb2} and references given there, and the paper of 
Symanzik on the small-distance behaviour analysis
\cite{Sym}).
\end{remark} 

The hypothesis of ETCR has been in serious doubt for a long time, see, 
\emph{e.g.}, the remarks in \cite{StreWight}, p.~101. Its validity has been 
tested \cite{WFW} in a large class of models in two-dimensional space-time; 
the Thirring model \cite{Thirrm}, the Schroer model \cite{Schrm}, the 
Thirring-Wess model of vector mesons interacting with zero-mass 
fermions (see \cite{ThirrWessm}, \cite{DubTar}), and the Schwinger model \cite{Schw}, 
using, throughout, the formulation of Klaiber \cite{Klai} for the Thirring model, 
and its extension to the other models by Lowenstein and Swieca \cite{LSwi} - 
for the Schwinger model, the previously mentioned noncovariant gauge ``$\alpha = \sqrt{\pi}$'' 
was adopted. Except for the Schwinger model, whose special canonical structure is due to 
its equivalence (in an irreducible sector) to a theory of a free scalar field of positive mass, 
the quantity
	\begin{equation}
		\{ \psi(x),\psi(y) \} - \langle \Omega,\{\psi(x),\psi(y)\}\Omega \rangle   
		 \cdot \mathbb{1} \; , 
	\end{equation}
where the $\psi$'s are the interacting fermi fields in the models
and $\{ \, . \, , \, . \,  \}$ denotes 
the anti-commutator and $\Omega$ denotes the vacuum, do not exist in the equal 
time limit as operator-valued distributions, for a certain range of coupling constants. 
Two different definitions of the equal time limit were used and compared, one of them 
due to Schroer and Stichel \cite{SchrSti}. 

Thus, the ETCR is definitely not true in general. In perturbation theory, 
$Z_{3}(\Lambda)$ (\cite{Weinb1}, p.~462) satisfies  $Z \le 1$ (see \eqref{(23.2)}) 
for all ultraviolet 
cutoffs~$\Lambda$, but it is just this condition which relies on the ETCR assumption 
and is not expected to be generally valid. In the limit $\Lambda \to \infty$, however, 
$Z_{3}(\Lambda)$  tends to~$- \infty$ and hence
violates \eqref{(14)} maximally.  In fact, 
\eqref{(14)} is violated  even for finite, sufficiently large $\Lambda$.

The models also provide examples of the validity of the singularity hypothesis (for the currents, 
analogous assertions hold if the commutator is used in place of the anti-commutator). 

By a.)~of 
Proposition~\ref{prop:2} it follows that $Z= 0$ (see \eqref{(17)}) need not be valid, 
even if the singularity hypothesis is valid. The question may now be posed: 
what is then the physical meaning of $Z= 0$, or, alternatively: 
when is $Z= 0$ valid?

In order to try to answer this question, we recall that, in the presence of massless 
particles (photons), Buchholz \cite{Buch} used Gauss' law to show, in a beautiful
paper, that the discrete spectrum of the mass operator
	\begin{equation}
		P_{\sigma} P^{\sigma} = M^{2} = P_{0}^{2}-\vec{P}^{2}
		\label{(25)}
	\end{equation}
is empty. Above, $P^{0}$ is the generator of time translations in the physical 
representation, \emph{i.e.}, the physical hamiltonian $H$, and $\vec{P}$ is the physical 
momentum. This fact is interpreted as a confirmation of the phenomenon that 
particles carrying an electric charge are accompanied by clouds of soft photons.

Buchholz begins by formulating adequate assumptions, viz., given that one wishes to 
determine the electric charge of a physical state $\Phi$ with the help of Gauss' law, 
the space-like asymptotic electromagnetic field of this state must i.) be measurable 
and ii) with sufficient precision. Let
	\begin{equation}
		F_{\mu,\nu}(\phi_{R}) \equiv \int d^{4}x \; \frac{\phi(x/R)}{R^{2}} \, F_{\mu,\nu}(x) \; . 
		\label{(26)}
	\end{equation}
Here, $\phi$ is an arbitrary real test function with compact support in the space-like 
complement of the origin in Minkowski space and $R>0$  is 
a scaling parameter. As $R$ 
increases the electromagnetic field in \eqref{(26)} is averaged over regions whose diameter 
and space-like distance from the origin grows like $R$; this average is 
rescaled by the scaling dimension of the field. He expresses i.) for a state $\Phi$ in the form
	\begin{equation}
		\lim_{R \to \infty} \langle \Phi, F_{\mu,\nu}(\phi_{R}) \Phi \rangle = f_{\mu,\nu}(\phi) \; ; 
		\label{(27)}
	\end{equation}
and ii.) as the boundedness of the mean square deviation of \eqref{(26)}:
	\begin{equation}
		\limsup_{R \to \infty} \; \bigl\| \bigl( F_{\mu,\nu}(\phi_{R})
		-f_{\mu,\nu}(\phi) \mathbb{1} \bigr)\Phi \bigr\|^{2}<\infty \; . 
		\label{(28)}
	\end{equation}
It is also assumed that, for the states $\Phi$, Gauss' law 
		\begin{equation}
			\langle \Phi, j_{\mu} \Phi \rangle= \;
                        \langle \Phi, \partial^{\nu}F_{\nu,\mu} \Phi \rangle
			\label{(29)}
		\end{equation}
holds in the sense of distributions on ${\cal S}(\mathbf{R}^{n})$, which allows 
to determine the electric charge of $\Phi$ by
	\begin{equation}
		\langle \Phi, Q \Phi \rangle = \lim_{R \to \infty} \int d^{4}x \, \frac{\chi(x/R)}{R} \;  
		\langle \Phi,j_{0}(x)\Phi \rangle =f_{i0}(\partial^{i}\chi) \; , 
		\label{(30)}
	\end{equation}
where $f_{\mu,\nu}$ is the functional \eqref{(26)} and $\chi$ is any test function, whose 
spatial derivatives $\partial^{i}\chi$ have support in the region $\{x \mid x^{2} < 0\}$ 
and which is normalized such that $\int d^{4}x \, \chi(x)\delta^{3}(x) = 1$. There is 
a further technical assumption
in Buchholz's paper, in order to avoid domain questions, for which we must refer to his paper 
(see his footnote 1), but omit from the statement of the following result,
which we chose to formulate as a theorem due to the central role it plays in our considerations.
We adjourn, however, the discussion of the 
applicability of the assumptions to concrete models as $qed_{1+3}$. 

\begin{theorem}[Buchholz~\cite{Buch}]
\label{th:4} Let $\Phi$ be a state in the complement of the vacuum 
satisfying Gauss' law in the sense of \eqref{(29)} 
as well as \eqref{(27)} and~\eqref{(28)}. Suppose, in addition, that 
        \begin{equation}
		P_{\sigma}P^{\sigma} \Phi = m^{2} \Phi
		\label{(32)}
	\end{equation}
for some $m^{2} \ge 0$, then
	\begin{equation}
		f_{\mu,\nu} = 0 \; ; 
		\label{(33)}
	\end{equation}
\emph{i.e.},  according to  \eqref{(26)} the state $\Phi$ is chargeless. 
\end{theorem}

\begin{remark}
In the above theorem, the fact that, in a Poincar\'e invariant quantum 
(field) theory, the joint energy-momentum spectrum restricted to any
carrier subspace of an irreducible representation of the (restricted) 
Poincar\'e group of mass $m \ge 0$ (lying therefore in the orthogonal 
complement of the vacuum vector), is absolutely continuous, plays a major
role. This fact was proved by Maison~\cite{Maison} in an important, but today 
somewhat forgotten paper. An analogous statement is proved to hold 
for a Galilei-invariant quantum (field) theory in the same reference.
The aforementioned results are a consequence of special properties of
the irreducible representations of the Poincar\'e and Galilei gropups,
and are therefore independent of locality. 
\end{remark}

When endeavouring to apply Theorem~\ref{th:4} to concrete models such as $qed_{1+3}$, 
problems similar to those occurring in connection with the charge superselection 
rule \cite{StrWight} arise. The most obvious one is that Gauss' law 
\eqref{(29)} is only expected to be 
valid (as an operator equation in the distributional sense) in non-covariant 
gauges (see \eqref{(8)}), the Coulomb gauge in the case of $qed_{1+3}$,  but not in 
covariant gauges \cite{StrWight}. We adopted, however, the option of  staying
with the Coulomb gauge and defining the theory in terms of  the 
$n$-point Wightman functions of observable  fields, 
\emph{i.e.}, gauge-invariant fields, thus maintaining 
Hilbert-space positivity. Within this framework, the hypotheses of Theorem~\ref{th:4} are 
in consonance with the requirements of Wightman's theory~\cite{StreWight}, 
and Buchholz's theorem should be applicable to $qed_{1+3}$. 
Recalling~\eqref{(19.1)} and~\eqref{(20.1)}, we also have the following result: 

\begin{corollary}
\label{cor:3} 
The assertion of Theorem~\ref{th:4} is equivalent to the assertion of the following conditions:
	\begin{equation}
		Z_{3} = 0
		\label{(40)}
	\end{equation}
and
	\begin{equation}
		Z_{2} = 0 \; . 
		\label{(41)}
	\end{equation}
\end{corollary}
  
It is interesting to recall, in connection with Corollary~\ref{cor:3}, that in \cite{JaWre2}, 
\emph{both} the photon field  and the electron-positron field are ``dressed''.

\section{Weinberg's criterion $Z=0$ for composite (unstable) particles}

We now turn to a different but, surprisingly perhaps, related subject, that 
of \emph{composite} or \emph{unstable particles}.  Classically speaking, these are 
particles whose field does not appear in the 
Lagrangian (\cite{Weinb1}, p.~461, \cite{Weinb4}). 

We follow \cite{Weinb1}, p.~461, but with some modifications: 
1) Weinberg formulated his criterion in terms of functional integrals, we prefer to formulate it 
in terms of classical field theory; 2)~some details are somewhat simplified and / or stated 
in a different form.

Turning to scalar fields for simplicity, we consider the case of a scalar particle $C$, 
of mass $m_{C}$, which may decay into a set of two (for simplicity) stable particles, each 
of mass $m$. We have energy conservation in the rest frame of $C$,  \emph{i.e.}, 
	\[
		m_{C}=\sum_{i=1}^{2} \sqrt{{\vec{q}_{i}\,}^{2}+m_{i}^{2}} 
		\ge \sum_{i=1}^{2} m_{i} \; , 
	\]
with $m_{i}=m,i=1,2$, and $\vec{q}_{i}$  the 
momenta of the two particles in the rest frame of $C$:
	\begin{equation}
		m_{C} > 2m  \; . 
		\label{(42)}
\end{equation}

We now assume that once can start from a classical Lagrangian field theory
and use canonical quantization to construct
a quantum field theory of an unstable (composite) particle $C$ with cutoffs; 
and that the limit exists, when these cutoffs are removed. 

\bigskip

\textbf{Assumption A} Proposition 3.2, together with (22), remain valid for the quantum field $C$ of a
composite (unstable) particle. The classical limit of the theory is assumed to yield a classical
scalar field, with a free Lagrangian density corresponding to mass $m_{C}$. Equation (24) is 
assumed to define the physical fields and their respective classical limits whenever $0<Z<\infty$.

\bigskip

The classical limit of quantum fields is reviewed in Duncan's monograph \cite{Duncan}.
As in the massless case, (24) does not define physical fields if $Z=0$: they have to be defined
differently in that case (see the remarks in 
the conclusion the massless case). In the unstable
case, see the remarks by Landsman in \cite{Landsman}, pg.~157: 
the LSZ condition $C \to 0$ seem to 
suggest that the $C$ particle ``dissipates into nothingness''. In particular, 
(24) is nonsensical
when $Z_{C}=0$, because it suggests that, either $C \equiv 0$ 
when $C_{ren}$ is well-defined, or
that it is inconclusive otherwise.    

We now turn to the classical picture of an unstable particle of mass $m$, described by a scalar field $A(x)$, 
and let us assume that the field $C(x)$ of particle $C$ is a functional of the 
field $A(x)$, \emph{i.e.},
	\begin{equation}
		C(x) = F(A(x)) \; , 
		\label{(43)}
	\end{equation}
with $F$ some, for the present purpose, unspecified function (which must satisfy some 
regularity conditions; see  below). 
\eqref{(43)} expresses the fact that the theory is equivalent to one 
in which $C(x)$ does not appear, \emph{i.e.}, it represents a particle composed of two 
stable $A$ particles, the latter's field $A(x)$ being the only fundamental constituent of the 
Lagrangian density. In the above argument, the fields in \eqref{(43)} are the physical (asymptotic) 
fields, which define the particle structure in the classical limit, in conformance with the Assumption.

\goodbreak
The equations \eqref{(43)} are the Euler-Lagrange equations associated to the 
classical Lagrangian density
	\begin{equation}
		{\cal L}_{I}(x) \equiv \frac{1}{2} \Bigl( C(x)-F(A(x)) \Bigr)^{2} \; . 
		\label{(44)}
	\end{equation}
In agreement with the assumption, the free Lagrangian density ${\cal L}_{0}$ for a scalar field 
is defined in terms of 
a ``bare'' or unrenormalized field $C_{0}$, given by
	\begin{equation}
		C_{0}(x) = Z_{C}^{1/2} C(x)
		\label{(45)}
		\quad \text{where} \quad 
			0 < Z_{C} < \infty \; . 
	\end{equation}
We have, therefore,
	\begin{equation}
		{\cal L}_{0}(x) = \frac{1}{2}\, \partial_{\mu}C_{0}(x)\partial^{\mu} \, C_{0}(x)
		-\frac{1}{2} \, m_{C}^{2} \, C_{0}(x)^{2} \; . 
		\label{(46)}
	\end{equation}
Equ.~\eqref{(43)} motivates the proposal of the following Ansatz for 
the full classical Lagrangian density associated to particle $C$:
	\begin{equation}
		{\cal L}(x) = {\cal L}_{0}(x) + {\cal L}_{I}(x) \; . 
		\label{(47)}
	\end{equation}
By \eqref{(45)}, the Euler-Lagrange equations for ${\cal L}$ are
	\begin{equation}
		Z_{C}^{1/2} (\square \, C(x) +m_{C}^{2} \, C(x))-
		\Bigl(C(x)-F(A(x)) \Bigr)=0 \; . 
		\label{(48)}
	\end{equation}
We have now the following result:

\begin{proposition}
\label{prop:3}
Let $A$ and $C$ belong to $C^{1}(\mathbb{R}^{4};\mathbb{R})$ and 
$F \in C^{1}(\mathbb{R};\mathbb{R})$.
If, furthermore, the Assumption A holds, then a composite field $C(\, . \, )$ 
of the form~\eqref{(43)} is a free field of mass $m_{C}$.
\end{proposition}

\begin{proof}
The regularity assumptions imply that the Lagrange density 
${\cal L} \in C^{2}(\mathbb{R}^{4} \times \mathbb{R} \times \mathbb{R}^{4})$ as required in 
classical field theory (\cite{BlanBru}, pg.~110): equation~(66) defines thus an equality
between continuous functions, and therefore, 
if~\eqref{(43)} holds, 
and $Z_{C} \ne 0$, from the assumption, we obtain 
	\[
		\square \, C(x) +m_{C}^{2} \, C(x) = 0 \; ; 
	\]
that is, $C(\, . \, ) $ is a free field of mass $m_{C}$.
\end{proof}

\begin{corollary}
\label{cor:6.2}
If $C$ is not a free field of mass $m_{C}$, \emph{i.e.}, if 
it is an ``interacting field'', then
the corresponding field strength renormalization constant $Z_{C}$ is zero, \emph{i.e.},
	\begin{equation}
		Z_{C} = 0 \; . 
		\label{(49)}
	\end{equation}
\end{corollary}

\begin{proof} This follows from the assumption, whereby \eqref{(49)} is 
the only remaining option due to (22).
\end{proof}

The idea of the criterion given by Corollary~\ref{cor:6.2}, 
which we call \emph{Weinberg's criterion},
is the relation of compositeness with \emph{functional dependence} (61), which is quite different
from the other approaches in the literature, to which refer to \cite{Landsman} (also for a very
comprehensive discussion and references).

In order to check in a model 
that either $Z_{C}=0$ when \eqref{(42)} holds or that  $0<Z_{C}<\infty$
in the stable case $m_{C} < 2m$, we are beset with the difficulty to 
obtain information 
on the two-point function. An exception are those rare cases
in which the (Fock) zero particle state is persistent (\cite{HeppB}), \emph{i.e.}, 
Lee-type models.

In fact, surprisingly, there exists a quantum model of Lee type of a composite (unstable) particle, 
satisfying \eqref{(42)}, where \eqref{(49)}
was indeed found, that of Araki et al.~\cite{AMKG}. Unfortunately, however, the (heuristic) results
in \cite{AMKG} have one major defect: their model contains ``ghosts''.

There exists, however, a ghostless version of 
the model treated heuristically in \cite{AMKG}, with the correct kinematics, due to 
Hepp (Theorem 3.4, pg.~54, of \cite{HeppB}). In this version, the masses in \eqref{(42)},
which are, of course, renormalized masses, may be determined
rigorously from the selfadjointness of the renormalized Hamiltonian.
It is an interesting open problem to carry out this investigation in detail.

For atomic resonances, examples exist \cite{HuSpohn}, but the situation 
there is entirely different from the particle case, because, for a bound-state 
problem, exponential decay of the wave-functions at 
infinity provides a natural ultraviolet cutoff. For particles, the issue is to remove 
the ultraviolet cutoff, keeping the physical quantities (decay rates) fixed.

\section{Conclusion}

In conclusion, \eqref{(40)}, \eqref{(41)} and \eqref{(49)} assert, by 
totally different arguments, that the physical reason for the occurrence 
of \eqref{(17)} is the presence of photons 
or composite (unstable) particles in the theory. As Weinberg 
remarks (\cite{Weinb1}, p.~461), 
\eqref{(49)} signals that the particle is ``maximally coupled to its constituents''. 
In the case of theories 
of massless photons, the ``dressing'' of electrons by photon clouds or of photons by  
electron-positron clouds 
may also be viewed similarly: each particle loses its identity as a single object. 
In the case of unstable 
particles, this identity is not recovered for asymptotic times, in the case 
of stable particles the ``clouds'' 
hopefully disappear asymptotically in the charge sectors, allowing the 
construction of a scattering theory even
under the assumptions \eqref{(40)} and \eqref{(41)} - the parameters of 
the free particles being determined
by recourse to the nonrelativistic limit of the theory. Alternatively, and most 
interestingly, relaxing condition~\eqref{(28)}, a scattering theory for QED 
was constructed by Alazzawi and 
Dybalski \cite{ADyb} using entirely 
new ideas and methods, related to the concept of superselection sector 
introduced in~\cite{BuchRo}.

It should be remarked, however, that the crucial test of nonperturbative 
relativistic gauge theory will
be to furnish a measurable number. For QED the most famous ones are 
not related to scattering theory,
but to spectral properties of certain Hamiltonians, for the electron g-factor 
see \cite{HFG}, for the Lamb shift
see \cite{JJS},\cite{DaNu}. In the latter reference,  the importance of the natural 
line shape for the Lamb shift - an unstable bound state problem -  is discussed. 
So far, the only experimental consequence
of a non-perturbative quantum field theoretic model concerns the Thirring 
model \cite{Thirrm}, which, in its
lattice version, the Luttinger model (whose first correct solution is due to 
Lieb and Mattis \cite{LiMa}) yields
a well-established picture of conductivity along one-dimensional 
quantum wires \cite{MaMa}. 
Hence,  it would be very important to find the ground state energy of 
positronium using the methods of \cite{JaWre2}. 

The picture offered in this paper is, of course, radically different from that 
of free quantum fields, which do 
satisfy the ETCR. However, it turns out that, since the photon has a hadronic 
component \cite{Berger} and all but the lightest particles are unstable,  the condition $Z=0$  
is a universal condition in particle physics; but not for the reasons hitherto assumed!  

\section{Acknowledgements} 
We should like to thank Erhard Seiler for valuable discussions.

\end{document}